\documentclass[cleveref,numberwithinsect]{lipics-v2021}

\usepackage{tikz}

\usepackage{amsmath}
\usepackage{stmaryrd}
\usepackage{mathtools}
\usepackage{amssymb}
\usepackage{amsthm}
\usepackage{complexity}
\usepackage{todonotes}
\usepackage{bm}
\usepackage{newfile}
\usepackage{xspace}

\newcommand{\ratnorm}[1]{\|#1\|_{\mathsf{frac}}}
\newcommand{\norm}[1]{\|#1\|_\infty}
\newcommand{\bx}{\bm{x}}
\newcommand{\bd}{\bm{d}}
\newcommand{\bc}{\bm{c}}
\newcommand{\bb}{\bm{b}}
\newcommand{\ba}{\bm{a}}
\newcommand{\bz}{\bm{z}}
\newcommand{\by}{\bm{y}}
\newcommand{\br}{\bm{r}}
\newcommand{\bw}{\bm{w}}
\newcommand{\bu}{\bm{u}}
\newcommand{\bv}{\bm{v}}
\newcommand{\bzero}{\bm{0}}
\DeclareMathOperator{\cone}{cone}
\DeclareMathOperator{\convhull}{conv.hull}
\DeclareMathOperator{\linspace}{lin.space}
\DeclareMathOperator{\charcone}{char.cone}
\newcommand{\Z}{\mathbb{Z}}
\newcommand{\N}{\mathbb{N}}
\newcommand{\Q}{\mathbb{Q}}
\renewcommand{\R}{\mathbb{R}}

\newcommand{\BEPA}{\ensuremath{\exists^{\le}\mathrm{PA}}\xspace}
\newcommand{\EPA}{\ensuremath{\exists\mathrm{PA}}\xspace}

\EventEditors{Karl Bringmann, Martin Grohe, Gabriele Puppis, and Ola Svensson}
\EventNoEds{4}
\EventLongTitle{51st International Colloquium on Automata, Languages, and Programming (ICALP 2024)}
\EventShortTitle{ICALP 2024}
\EventAcronym{ICALP}
\EventYear{2024}
\EventDate{July 8--12, 2024}
\EventLocation{Tallinn, Estonia}
\EventLogo{}
\SeriesVolume{297}
\ArticleNo{153}

\title{An efficient quantifier elimination procedure for {P}resburger arithmetic}

\author{Christoph Haase}{Department of Computer Science, University of Oxford, UK}{christoph.haase@cs.ox.ac.uk}{https://orcid.org/0000-0002-5452-936X}{}

\author{Shankara Narayanan Krishna}{Department of Computer Science \& Engineering, IIT Bombay, India}{krishnas@cse.iitb.ac.in}{https://orcid.org/0000-0003-0925-398X}{}

\author{Khushraj Madnani}{Max Planck Institute for Software Systems (MPI-SWS), Germany}{kmadnani@mpi-sws.org}{https://orcid.org/0000-0003-0629-3847}{}

\author{Om Swostik Mishra}{Department of Mathematics, IIT Bombay, India}{21b090022@iitb.ac.in}{https://orcid.org/0009-0001-6858-6605}{}

\author{Georg Zetzsche}{Max Planck Institute for Software Systems (MPI-SWS), Germany}{georg@mpi-sws.org}{https://orcid.org/0000-0002-6421-4388}{}
\authorrunning{Haase, Krishna, Madnani, Mishra, and Zetzsche}

\Copyright{Christoph Haase, Shankara Narayanan Krishna, Khushraj Madnani, Om Swostik Mishra, and Georg Zetzsche}


\category{Track B: Automata, Logic, Semantics, and Theory of Programming}

\nolinenumbers


\hideLIPIcs

\begin{document}

\thispagestyle{empty} %

\maketitle

\funding{\flag[3cm]{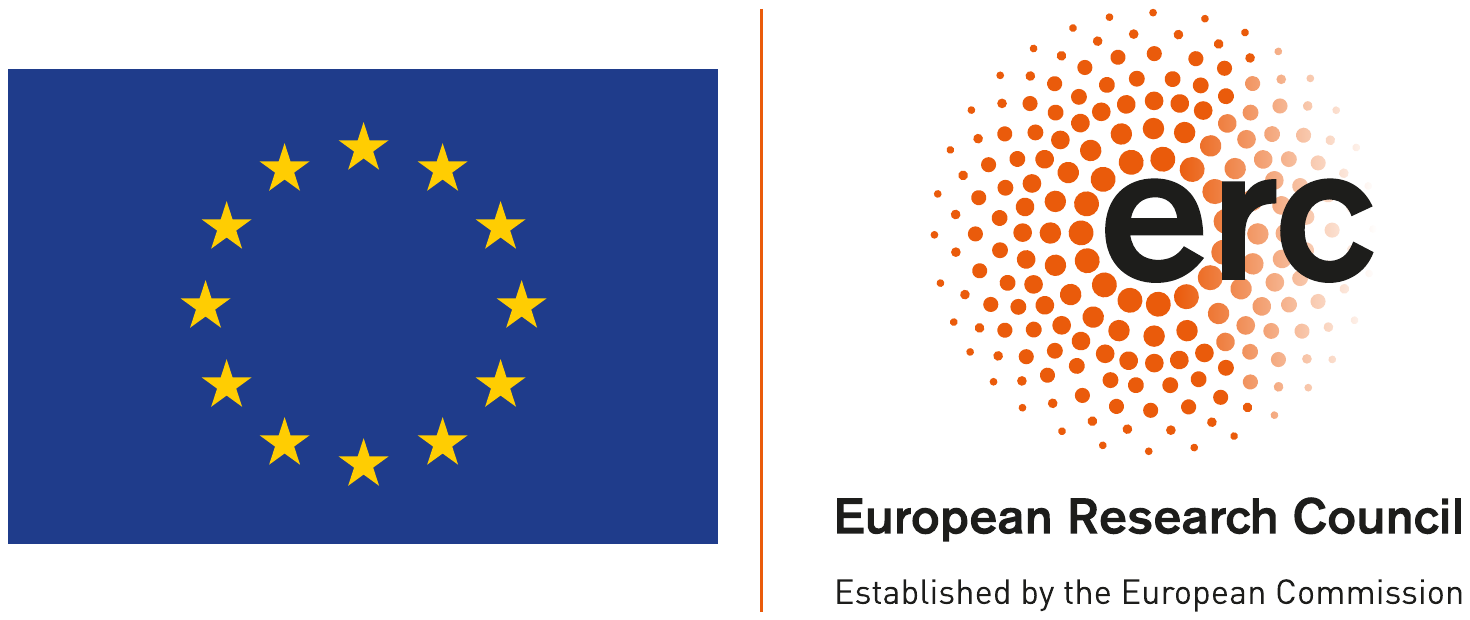}Funded by the European Union (ERC,
  FINABIS, 101077902). Views and opinions expressed are however those
  of the author(s) only and do not necessarily reflect those of the
  European Union or the European Research Council Executive
  Agency. Neither the European Union nor the granting authority can be
  held responsible for them. Christoph Haase is supported by the
  European Research Council (ERC) under the European Union’s Horizon
  2020 research and innovation programme (Grant agreement No. 852769,
  ARiAT).}

\acknowledgements{We are grateful to 
(i)~Pascal Bergstr\"a{\ss}er, Moses Ganardi, and Anthony W. Lin for discussions about Weispfenning's lower bound, 
(ii)~Pascal Baumann, Eren Keskin, Roland Meyer for discussions on polyhedra, and 
(iii)~Anthony W. Lin and Matthew Hague for explaining some aspects of their results on monadic decomposability.}

\begin{abstract}
  All known quantifier elimination procedures for Presburger arithmetic require
	doubly exponential time for eliminating a single block of existentially
	quantified variables. It has even been claimed in the literature that
	this upper bound is tight. We observe that this claim is incorrect and
	develop, as the main result of this paper, a quantifier elimination
	procedure eliminating a block of existentially quantified variables in
	singly exponential time. As corollaries, we can establish the precise
	complexity of numerous problems. Examples include deciding (i)~monadic
	decomposability for existential formulas, (ii)~whether an existential
	formula defines a well-quasi ordering or, more generally, (iii)~certain
	formulas of Presburger arithmetic with Ramsey quantifiers. Moreover,
	despite the exponential blowup, our procedure shows that under mild
	assumptions, even $\NP$ upper bounds for decision problems about
	quantifier-free formulas can be transferred to existential formulas.
	The technical basis of our results is a kind of small model property
	for parametric integer programming that generalizes the seminal results
	by von zur Gathen and Sieveking on small integer points in convex
	polytopes.
\end{abstract}

\section{Introduction}

Presburger arithmetic is the first-order theory of the integers with
addition and order. This theory was shown decidable by Moj\.zesz
Presburger in 1929~\cite{Pre29} by establishing a quantifier elimination procedure
in the extended structure additionally consisting of infinitely many
predicates $m \mid {\cdot}$ for all integers $m>0$, asserting
divisibility by a constant. Recall that a logical theory $T$ admits
quantifier elimination whenever for any formula
$\Phi(y_1,\ldots,y_k)\equiv \exists x\, \varphi(x,y_1,\ldots,y_k)$
with $\varphi$ being quantifier free there is a computable
quantifier-free formula $\Psi(y_1,\ldots,y_k)$ such that
$\Phi \leftrightarrow \Psi$ is a tautology in $T$. Presburger's
quantifier elimination procedure has non-elementary running time. In
the early 1970s, Cooper~\cite{Coo72} developed an improved version of
Presburger's procedure, which was later shown to run in triply
exponential time~\cite{Oppen78}. Ever since, various other quantifier
elimination procedures have been established and analyzed, especially
for fragments of Presburger arithmetic with a fixed number of
quantifier alternations, see e.g.~\cite{RL78,Wei90}. Weispfenning~\cite{DBLP:conf/issac/Weispfenning97}
analyzed lower bounds for quantifier-elimination procedures and showed
that, assuming unary encoding of numbers, \emph{any} quantifier
elimination procedure requires triply exponential time. In the same
paper, Weispfenning also claims that any algorithm eliminating a single block of
existential quantifiers inherently requires \emph{doubly} exponential
time~\cite[p. 50]{DBLP:conf/issac/Weispfenning97}.

The main contribution of this paper is to develop a quantifier
elimination procedure for Presburger arithmetic that eliminates a
block of existentially quantified variables in \emph{singly}
exponential time. This, of course, contradicts Weispfenning's claim,
which actually turns out to be incorrect as we point out in detail
in \Cref{appendix-weispfenning}. The key technical insight underlying
our procedure is a kind of small model property for parametric integer
programming. Given an integer matrix $A \in \Z^{\ell \times n}$ and
$\bb\in \Z^\ell$, recall that integer programming is to decide whether
 there is some $\bx\in \Z^n$ such that $A \bx \le \bb$. It is
well-known by the work of von zur Gathen and Sieveking~\cite{vzGS78},
and Borosh and Treybig~\cite{BT76}, that if such an $\bx$ exists then
there is one whose bit length is polynomially bounded in the bit
lengths of $A$ and $\bb$. In this paper, we refer to the situation in
which $\bb$ is not fixed and provided as a parameter as \emph{parametric
integer programming}. Our main technical result states that, in this
setting, if $A\bx \le \bb$ has a solution for a given
$\bb \in \Z^\ell$ then there are $D\in \Q^{n\times \ell}$ and
$\bd \in \Q^n$, both of bit length polynomial in the bit length of
$A$, such that $\bx = D\bb + \bd$ is integral and also a
solution. Observe that there is only an exponential number (in the bit
length of $A$) of possible choices for $D$ and $\bd$. Eliminating a
block of variables $\bx$ from a system of linear inequalities thus becomes
easy: we have that $A \bx \le B \by + \bc$ is equivalent to the
disjunction of systems of the form $A(D(B \by + \bc)+\bd) \le B \by + \bc$
for all $D$ and $\bd$ of bit length polynomial in $A$. Using standard
arguments, this approach can then be turned into a quantifier
elimination procedure that eliminates a block of existentially
quantified variables in exponential time.

\section{Preliminaries}

Throughout this paper, all vectors $\bz$ are treated as column vectors unless
mentioned otherwise.  For a vector $\bx\in\Q^n$, let $\norm{\bx}$ be the
maximal absolute value of all components of $\bx$.  Moreover, let
$\ratnorm{\bx}$ be the maximal absolute value of all \emph{numerators} and
\emph{denominators} of components in $\bx$.  The latter is important for
representations: Note that a vector $\bx\in\Q^n$ with $\ratnorm{\bx}\le m$ can
be represented using $O(n\log m)$ bits.  We use analogous notations $\norm{A}$
and $\ratnorm{A}$ for matrices $A$.  We will sometimes refer to the
\emph{Hadamard inequality}~\cite{Had93}, which implies that for a square matrix
$A\in\Z^{n\times n}$, we have $|\det(A)|\le (n\cdot\norm{A})^n$. In particular,
the determinant of $A$ is at most exponential in the maximal absolute value of
entries of $A$.

\subparagraph{Presburger arithmetic} \emph{Presburger arithmetic} (PA) is the first-order theory of the structure $\langle \Z; +, <,0,1\rangle$. In order to enable quantifier elimination, we have to permit modulo constraints. Thus technically, we are working with the structure $\langle \Z;+,<,(\equiv_m)_{m\in\Z},0,1\rangle$, where $a\equiv_m b$ stands for $a\equiv b\bmod{m}$.
In our syntax, we allow atomic formulas of the forms $a_1x_1+\cdots+a_nx_n\le b$ (called \emph{linear inequalities}) or $a_1x_1+\cdots+a_nx_n\equiv b\bmod{m}$ (called \emph{modulo} or \emph{divisibility constraints}), where $x_1,\ldots,x_n$ are variables and $a_1,\ldots,a_n,b,m\in\Z$ are constants encoded in binary. A formula is \emph{quantifier-free} if it contains no quantifiers or, equivalently, is a Boolean combination of atomic formulas. Notice that conjunctions of linear inequalities can be written as systems of linear
inequalities $A\bx \le \bb$.

The \emph{size} of a PA Formula $\varphi$, denoted $|\varphi|$, is the number of letters used to write it down, where we assume all constants to be encoded in binary. (Sometimes, we say that a formula obeys a size bound even if constants are encoded in unary; but this will be stated explicitly).

\subparagraph{Fixed quantifier alternation fragments} The $\Sigma_k$ fragment of PA consists of formulas of the form $\exists \bm{u_1}\forall \bm{u_2}\dots Q_k\bm{u_k}\colon\varphi(\bm{u_1},\bm{u_2},\dots,\bm{u_k},\bz)$ where $\bm{u_i}$ is a vector of quantified variables, $\bz$ is a vector of free variables, $\varphi(\bm{u_1},\bm{u_2},\dots,\bm{u_k},\bz)$ is a quantifier free PA formula, and 
$Q_k$ denotes $\forall$ or $\exists$ depending on whether $k$ is even or odd respectively.
Similarly, the $\Pi_k$ fragment of PA consists of formulas of the form $\forall \bm{u_1}\exists \bm{u_2}\dots Q_k\bm{u_k}\colon\varphi(\bm{u_1},\bm{u_2},\dots,\bm{u_k},\bz)$ where $Q_k$ denotes $\forall$ or $\exists$ depending on whether $k$ is odd or even respectively. %

\subparagraph{Bounded existential Presburger arithmetic}
In addition to our quantifier elimination result, we shall prove a somewhat stronger version, which states that one can compute a compact representation of a quantifier-free formula in polynomial time.
As compact representations, we introduce a syntactic variant of existential Presburger arithmetic, which we call \emph{bounded existential Presburger arithmetic}, short $\BEPA$.
Essentially, \BEPA\ requires all quantifiers to be restricted to bounded intervals, but also permits polynomials over the quantified variables. Using standard methods, one can translate every formula in \BEPA in polynomial time into an \EPA\ formula. However, the converse is not obvious, and our main results states that this is possible. Syntactically, an \BEPA\ formula over free variables $y_1,\ldots,y_m$ is of the form 
\[ \exists^{\le k_1}x_1\cdots\exists^{\le k_n} x_n\colon \varphi, \]
where $x_1,\ldots,x_n$ are variables, each $k_i\in\N$ is a number given in binary, and $\varphi$ is a quantifier-free formula where every atom is of one of the forms: 
\begin{equation} \sum_{i=1}^m p_iy_i\le q\quad\text{or}\quad \sum_{i=1}^m p_iy_i\equiv r\bmod{q}, \label{bepa-atoms}\end{equation}
where $p_1,\ldots,p_m,q,r\in\Z[x_1,\ldots,x_n]$ are polynomials over the variables $x_1,\ldots,x_n$. Thus, where \EPA\ allows constant integral coefficients, \BEPA\ allows polynomials from $\Z[x_1,\ldots,x_n]$. The quantifiers $\exists^{\le k_i}x_i$ are interpreted as ``there exists $x_i\in\Z$ with $|x_i|\le k_i$''.

\begin{remark}
	\label{bepa-to-epa}
	Now indeed, a \BEPA\ formula can be converted in polynomial time into an \EPA\
formula: The bounded quantification is clearly expressible in \EPA.
The terms $p_iy_i$ and $q$ in \eqref{bepa-atoms} (recall $p_i$ and $q$ are polynomials  are from $\Z[x_1,\ldots,x_n]$) are also expressible, because
multiplication with exponentially bounded variables can be expressed using
polynomial-size \EPA\ formulas. This is because given a polynomial
$p\in\Z[x_1,\ldots,x_n]$ and a variable $y$, we can construct a
polynomial-size existential formula $\pi(x_1,\ldots,x_n,y,z)$ expressing
$z=p(x_1,\ldots,x_n)\cdot y \wedge |x_1|\le k_1\wedge\cdots \wedge|x_n|\le k_n$.
This, in turn, follows from the fact that given $\ell$ in unary, we can
construct an existential formula $\mu_\ell(x,y,z)$, of size linear in $\ell$, expressing $z=x\cdot y\wedge
|x|\le 2^\ell$ (see~\cite[Sec.~3.1]{DBLP:conf/csl/Haase14} or \cite[p.~7]{DBLP:conf/lics/HaaseZ19}). Thus, we can construct $\pi$ in \EPA\ by introducing a
variable for each subterm of $p$ (which can clearly all be bounded exponentially).
\end{remark}

\subparagraph*{Making $\BEPA$ formulas quantifier-free}\label{bepa-quantifier-free}
Moreover, a $\BEPA$ formula can easily be converted (in exponential time) into
an exponential-size quantifier-free formula: Just take an exponential
disjunction over all assignments of the existentially bounded variables
$x_1,\ldots,x_n$ and replace the variables by their values in all the atoms.
Thus, $\BEPA$ formulas can be regarded as compact representations of quantifier-free formulas.

\section{Main results}
Here, we state and discuss implications of the main result of this paper:
\begin{theorem}\label{bounded-existential}
	Given a formula of $\EPA$, we can construct in
	polynomial time an equivalent formula in $\BEPA$.
\end{theorem}
From \cref{bounded-existential}, we can deduce the following, since by the remark \cref{bepa-to-epa} in \cref{bepa-quantifier-free}, one can easily convert a \BEPA\ formula
into an exponential-sized quantifier-free formula.
\begin{corollary}\label{quantifier-elimination}
	Given a formula $\varphi$ in existential Presburger arithmetic, we can compute
	in exponential time an equivalent quantifier-free formula $\psi$ of size
	exponential in $\varphi$. Moreover, all constants in $\psi$ are encoded in
	unary.
\end{corollary}
In \cref{exponential-lower-bound}, we will see that an exponential blowup
cannot be avoided when eliminating a block of existential quantifiers,
even if we allow constants to be encoded in binary in the quantifier-free formula.

There are several applications of
\cref{bounded-existential,quantifier-elimination}. The most obvious type of
applications are those, where, for every problem\footnote{To be precise: Every
	\emph{semantic} problem, meaning one that only depends on the set defined by
	the input formula. } that is in $\NP$ (resp.\ $\coNP$) for quantifier-free
formulas, the same problem belongs to $\NEXP$ (resp.\ $\coNEXP$) for
existential formulas. Oftentimes, this yields optimal complexity. We mention
some examples.

A direct consequence of \cref{quantifier-elimination} (and the $\NP$ membership of the quantifier free fragment of PA) is the following.
\begin{corollary}\label{nexp-sigma2}
	The $\Sigma_2$-fragment of Presburger arithmetic belongs to $\NEXP$.
\end{corollary}
The $\NEXP$ upper bound is known and was shown by Haase~\cite[Thm.~1]{DBLP:conf/csl/Haase14}. In fact, the $\Sigma_2$-fragment is known to be
$\NEXP$-complete: An $\NEXP$ lower bound was shown much earlier by
Gr\"{a}del~\cite{DBLP:journals/apal/Gradel89}, already for the
$\exists\forall^*$-fragment.

\newcommand{\ram}{\exists^{\mathsf{ram}}}
\subparagraph{Ramsey quantifiers}
In fact, combining \cref{quantifier-elimination} with the results from \cite{DBLP:journals/pacmpl/BergstrasserGLZ24}, we can strengthen \cref{nexp-sigma2}. The \emph{Ramsey quantifier} $\ram$ states the existence of infinite (directed) cliques. More precisely, if $\varphi(\bx,\by,\bz)$ is a Presburger formula where $\bx$ and $\by$ are vectors of $n$ variables each, then $\ram (\bx,\by)\colon\varphi(\bx,\by,\bz)$ is satisfied for $\bz$ if and only if there exists an infinite sequence $\ba_1,\ba_2,\ldots\in\Z^n$ of pairwise distinct vectors with $\varphi(\ba_i,\ba_j,\bz)$ for every $i<j$. As mentioned in \cite{DBLP:journals/pacmpl/BergstrasserGLZ24}, Ramsey quantifiers can be applied to deciding liveness properties, deciding monadic decomposability (see below), and deciding whether a formula defines a well-quasi-ordering (see below).

In \cite[Thm.~5.1]{DBLP:journals/pacmpl/BergstrasserGLZ24}, it is shown that if $\varphi(\bx,\by,\bz)$ is an \EPA\ formula, then one can compute in polynomial time an \EPA\ formula $\varphi'(\bz)$ equivalent to $\ram (\bx,\by,\bz)\colon \varphi(\bx,\by,\bz)$. 
\begin{corollary}\label{ramsey}
	Given a $\Sigma_2$-formula $\varphi(\bx,\by,\bz)$, we can construct an exponential-size \EPA\ formula equivalent to $\ram (\bx,\by)\colon\varphi(\bx,\by,\bz)$. In particular, deciding the truth of $\ram (\bx,\by)\colon\psi$ for
	$\Sigma_2$-formulas $\psi(\bx,\by)$ is $\NEXP$-complete.
\end{corollary}
Indeed, \cref{quantifier-elimination} lets us convert $\varphi(\bx,\by,\bz)$ into an exponential-size existential formula $\varphi'$, so that we can apply the above result of \cite{DBLP:journals/pacmpl/BergstrasserGLZ24} to the formula $\ram(\bx,\by)\colon \varphi'(\bx,\by,\bz)$, which results in an equivalent exponential \EPA\ formula. The $\NEXP$ lower bound in the second statement follows from $\NEXP$-hardness of the $\Sigma_2$-fragment and the fact that for a given $\Sigma_2$-formula $\chi$ without free variables, the statement $\ram(\bx,\by)\colon \chi\wedge \bx<\by$ is equivalent to $\chi$.

\subparagraph{Detecting WQOs}
A \emph{well-quasi-ordering} (WQO) is a reflexive and transitive ordering $(X,\le)$ such that for every sequence $x_1,x_2,\ldots\in X$, there are $i<j$ with $x_i\le x_j$.
Well-quasi-orderings are of paramount importance in the widely applied theory of well-structured transition systems~\cite{finkel87,abdulla00,finkel01}.
The problem of deciding whether a given Presburger formula $\varphi(\bx,\by)$ defines a WQO was recently raised by Finkel and Gupta~\cite{DBLP:conf/fsttcs/FinkelG19}, with the hope of establishing automatically that certain systems are well-structured. As observed in~\cite[Prop. 12]{DBLP:journals/corr/abs-1910-02736}, this problem reduces to evaluating Ramsey quantifiers, which is decidable by~\cite{DBLP:journals/bsl/Rubin08}.
Based on an $\NP$ algorithm for Ramsey quantifiers, it is shown in \cite[Sec.~8.3]{DBLP:journals/pacmpl/BergstrasserGLZ24} that given a quantifier-free formula $\varphi(\bx,\by)$, where $\bx$ and $\by$ are vectors of $n$ variables each, it is $\coNP$-complete whether the relation $R\subseteq\Z^n\times\Z^n$ defined by $\varphi$ is a WQO. Our results allow us to settle the complexity for existential formulas:
\begin{corollary}\label{detecting-wqos}
	Given an \EPA\ formula $\varphi$, it is
	$\coNEXP$-complete to decide whether $\varphi$ defines a
	WQO.
\end{corollary}
The upper bound follows directly from \cref{quantifier-elimination} and the fact that it is $\coNP$-complete to decide whether a given quantifier-free formula defines a well-quasi-ordering~\cite[Sec.~8.3]{DBLP:journals/pacmpl/BergstrasserGLZ24}. This yields a $\coNEXP$ procedure overall. It should be noted that \cref{detecting-wqos} can also be deduced from \cref{ramsey} (using the same idea as in \cite[Sec.~8.3]{DBLP:journals/pacmpl/BergstrasserGLZ24}). However, we find it instructive to demonstrate how quantifier elimination permits a direct transfer of the $\coNP$ algorithm as a black box. We show the $\coNEXP$ lower bound in \cref{lower-bounds}.

\subparagraph*{Monadic decomposability}
A Presburger formula is \emph{monadic} if each of its atoms contains at most one variable.  Moreover, we say that a Presburger formula $\varphi$ is \emph{monadically decomposable} if $\varphi$ is equivalent to a monadic Presburger formula. Motivated by the role monadic formulas play in constraint databases~\cite{GRS01,CDB-book}, Veanes, Bj{\o}rner, and Nachmanson, and Bereg recently raised the question of how to decide whether a given formula is monadically decomposable~\cite{DBLP:journals/jacm/VeanesBNB17}. For Presburger arithmetic, decidability follows from \cite[p.~1048]{ginsburg1966bounded} and for quantifier-free formulas, monadic decomposability was shown $\coNP$-complete in \cite[Thm.~1]{DBLP:conf/cade/HagueLRW20} (in~\cite[Cor.~8.1]{DBLP:journals/pacmpl/BergstrasserGLZ24}, the $\coNP$ upper bound is shown via Ramsey quantifiers). \cref{quantifier-elimination} allows us to settle the case of \EPA\ formulas.
\begin{corollary}\label{monadic-decomposability-epa}
	Monadic decomposability of \EPA\ formulas is $\coNEXP$-complete.
\end{corollary}
This is because given an \EPA\ formula, we can compute an exponential-sized quantifier-free formula and apply the existing $\coNP$ procedure, yielding a $\coNEXP$ upper bound overall. Again, the $\coNEXP$ upper bound could also be deduced from \cref{ramsey} (but this proof shows again how to transfer algorithms using quantifier elimination). The $\coNEXP$ lower bound follows the same idea as the $\coNP$ lower bound in \cite{DBLP:journals/pacmpl/BergstrasserGLZ24}, see \cref{lower-bounds}.

\newcommand{\fp}{\mathfrak{p}}
\newcommand{\QF}{\mathrm{QF}}
\subparagraph*{$\NP$ upper bounds}
In addition to new $\NEXP$ and $\coNEXP$ upper bounds, \cref{bounded-existential} can also be used to obtain $\NP$ upper bounds. Suppose we have a predicate $\fp$ on sets of integral vectors. That is, for each $S\subseteq\Z^m$ for some $m\in\N$, either $\fp(S)$ is true or not. We call this predicate \emph{admissible} if for any $m\in\N$, $S_1,S_2\subseteq\Z^m$, we have that $\fp(S_1\cup S_2)$ implies $\fp(S_1)$ or $\fp(S_2)$. Let us see some examples:
\begin{enumerate}[(i)]
	\item\label{fp-nonempty} The predicate $\fp$ with $\fp(S)$ if and only if $S\ne\emptyset$.
	\item\label{fp-infinite} The predicate $\fp$ with $\fp(S)$ if and only if $S$ is infinite.
	\item\label{fp-powers} The predicate $\fp$ with $\fp(S)$ if and only if $S\subseteq\Z$ and $S$ contains a power of $2$.
	\item\label{fp-ramsey} The predicate $\fp$ with $\fp(S)$ if and only if $S\subseteq\Z^{2k}$ and viewed as a relation $S\subseteq\Z^k\times\Z^k$, $S$ has an infinite clique.
	\item The predicate $\fp$ with $\fp(S)$ if and only if $S\subseteq\Z$ and $S$ contains infinitely many primes.
	\item The predicate $\fp$ with $\fp(S)$ if and only if $S\subseteq\Z^2$ and $S$ contains a pair $(x,2^x)$.
\end{enumerate}
For each such predicate, we consider the problem $\fp(\EPA)$:
\begin{description}
	\item[Input] An $\EPA$ formula $\varphi$ with $m$ free variables for some $m\in\N$.
	\item[Question] Does $\fp(S)$ hold, where $S\subseteq\Z^m$ is the set defined by $\varphi$?
\end{description}
Moreover, $\fp(\QF)$ is the restriction of the problem where the input formula $\varphi$ is quantifier-free.

For several of the examples above, it is known that $\fp(\EPA)$ is in $\NP$: For (\labelcref{fp-nonempty}) and (\labelcref{fp-infinite}), these are standard facts, and for (\labelcref{fp-powers}), this follows from $\NP$-completeness of existential Büchi arithmetic~\cite[Thm. 1]{DBLP:conf/lics/GuepinH019}. For (\labelcref{fp-ramsey}), this follows from the fact that Ramsey quantifiers can be evaluated in $\NP$~\cite[Thm~5.1]{DBLP:journals/pacmpl/BergstrasserGLZ24}. Our results imply that for proving $\NP$ upper bounds, we may always assume a quantifier-free input formula. This is perhaps surprising, because one might expect that for non-linear predicates, it is difficult to bound the quantified variables.
\begin{corollary}
	For every admissible predicate $\fp$, the problem $\fp(\EPA)$ is in $\NP$ if and only if $\fp(\QF)$ is in $\NP$.
\end{corollary}
Here, the ``only if'' direction is trivial, and the ``if'' direction follows
from \cref{bounded-existential}. This is because \cref{bounded-existential}
allows us to assume that $\varphi$ is given as a $\BEPA$ formula $\exists^{\le
	k_1}x_1\cdots\exists^{\le k_n}x_n\colon\psi(x_1,\ldots,x_n,y_1,\ldots,y_m)$.
Moreover, admissibility of $\fp$ implies that $\fp$ is satisfied for $\varphi$
if and only if there exists an assignment $(a_1,\ldots,a_n)$ for the bounded
variables such that the quantifier-free formula
$\psi(a_1,\ldots,a_n,y_1,\ldots,y_m)$ satisfies $\fp$. Thus, we can guess the
assignment (which occupies polynomially many bits) and run the $\NP$ algorithm
for quantifier-free formulas.

\section{Quantifier elimination}
In this section, we prove \cref{bounded-existential}.
The following is our main geometric ingredient.
\begin{proposition}\label{solution-affine-transformation}
Let  $A\in\Z^{\ell\times n}$ and $\bb\in\Z^{\ell}$, and let $\Delta$ be an upper bound on all absolute values of the subdeterminants of $A$. If the system $A\bx\le\bb$ has an integral solution, then it has an integral solution of the form $D\bb+\bd$, where $D\in\Q^{n\times\ell}$ and $\bd\in\Q^{n}$ with $\ratnorm{D}\le\Delta$ and $\ratnorm{\bd}\le n\Delta^2$.
\end{proposition}
Before we prove \cref{solution-affine-transformation}, let us see how it implies \cref{bounded-existential} and \cref{quantifier-elimination}. 

\subparagraph*{Proof of \cref{quantifier-elimination}}
While \cref{quantifier-elimination} follows from \cref{bounded-existential}, it
follows very directly from \cref{solution-affine-transformation} and the proof
is a good warm-up for the proof of \cref{bounded-existential}. Therefore, we
first derive \cref{quantifier-elimination} from
\cref{solution-affine-transformation}.  Suppose we are given a Presburger
formula $\exists\bx\colon\varphi(\bx,\by)$, where $\bx=(x_1,\ldots,x_n)$ and
$\by=(y_1,\ldots,y_m)$ are variables and $\varphi$ is quantifier-free.

It is well-known that divisibility constraints can be eliminated in favor of
existentially quantified variables, since $a\equiv b\bmod{m}$ if and only if
$\exists x\colon a-b=mx$.  Thus, we may assume that $\varphi$ contains no
divisibility constraints. Then, by moving all negations inwards and using the
standard equivalence $\neg(r\le t)\iff t+1\le r$, we may assume that $\varphi$
is a positive Boolean combination of atoms $\ba^\top\bx\le\bb^\top\by+c$, where
$\ba\in\Z^{n}$, $\bb\in\Z^{m}$, and $c\in\Z$.

By bringing $\varphi$ into DNF, we can write it as a disjunction of
exponentially many systems of inequalities of the form $A\bx\le B\by+\bc$,
where $A\in\Z^{\ell\times n}$, $B\in\Z^{\ell\times m}$, and $\bc\in\Z^{\ell}$.
Thus, it suffices to construct a quantifier-free formula for $\exists\bx\colon
A\bx\le B\by+\bc$. Let $\Delta$ be an upper bound for all absolute values of
subdeterminants of $A$. Since the transformation into DNF does not change
the appearing
constants, we have that $\Delta\le ((m+n)\norm{A})^{m+n}$ is at most exponential in the
size of the input formula.

According to \cref{solution-affine-transformation}, a vector $\bx$ with
$\varphi(\bx,\by)$ exists if and only if there exists a matrix $D\in\Q^{n\times
\ell}$ and $\bd\in\Q^{n}$ with $\ratnorm{D}\le\Delta$ and $\ratnorm{\bd}\le
n\Delta^2$ such that (i)~substituting $D(B\by+\bc)+\bd$ for $\bx$ satisfies
$A\bx\le B\by+\bc$ and also (ii)~the vector $D(B\by+\bc)+\bd$ is integral.
Therefore, the formula $\exists\bx\colon A\bx\le B\by+\bc$ is equivalent to 
\[ \bigvee_{(D,\bd)\in P} A(D(B\by+\bc)+\bd)\le B\by+\bc ~\wedge~ D(B\by+\bc)+\bd\in\Z^{n} \]
where $P$ is the set of all pairs $(D,\bd)$ with $D\in\Q^{n\times\ell}$, $\bd\in\Q^n$, $\ratnorm{D}\le\Delta$, and $\ratnorm{\bd}\le n\Delta^2$. Clearly, $P$ contains at most exponentially many elements. Moreover, note that the condition $D(B\by+\bc)+\bd\in\Z^{n}$ is a set of $n$ modulo constraints.

\subparagraph*{Proof of \cref{bounded-existential}}
The proof of \cref{bounded-existential} is similar to the above
construction---we just need to circumvent the exponential conversion into DNF.
We proceed as follows. 

As above, we are given a Presburger
formula $\exists\bx\colon\varphi(\bx,\by)$, where $\bx=(x_1,\ldots,x_n)$ and
$\by=(y_1,\ldots,y_m)$ are variables and $\varphi$ is quantifier-free. Moreover, we may assume that $\varphi$ contains no divisibility constraints and is a positive Boolean combination of atoms
$\ba^\top\bx\le\bb^\top\by+c$, where $\ba\in\Z^{n}$, $\bb\in\Z^{m}$, and
$c\in\Z$.

Let $\ba_i^\top\bx\le\bb_i^\top\by+c_i$ for $i=1,\ldots,\ell$ be the set of all atoms occurring in $\varphi$ and let $A\in\Z^{\ell\times n}$ be the matrix with rows $\ba_i^\top$ and $B\in\Z^{\ell\times m}$ be the matrix of rows $\bb_i^\top$, and let $\bc \in \Z^{\ell}$ be the (column) vector with entries $c_1,\ldots,c_\ell$.
Thus, our formula $\varphi$ consists of $\ell$ atoms, each of which is a row in the system of linear inequalities $A\bx\le B\by+\bc$. Let $\varphi'$ be the formula obtained from $\varphi$ by replacing the atom $\ba_i^\top\bx\le\bb_i^\top\by+c_i$ by $z_i=1$, where $z_i$, $i\in\{1,\ldots,\ell\}$, is a fresh variable for each of the $\ell$ atoms. 
Now let $\Delta$ be an upper bound on all absolute values of the
subdeterminants of $A$. Then $\Delta\le ((m+n)\norm{A})^{m+n}$ is at most exponential in the size of the input formula.
Consider the formula
\begin{multline}
\exists z_1,\ldots,z_\ell\in\{0,1\}\colon\quad \exists D\in\Q^{n\times\ell},\ratnorm{D}\le \Delta\colon \\
\exists \bd\in\Q^{n},~\ratnorm{\bd}\le n\Delta^2\colon\quad\varphi'~\wedge~ D(B\by+\bc)+\bd\in\Z^{n}~\wedge~\bigwedge_{i=1}^\ell \left(z_i=1\to \psi_i\right),\label{bepa-formula}
\end{multline}
where $\psi_i$ is the formula $\ba_i^\top (D(B\by+\bc)+\bd)\le\bb_i^\top\by+c_i$. Note that \eqref{bepa-formula} is expressible in \BEPA: We introduce (i)~one variable for each $z_i$, (ii)~two variables for each entry of $D$ (one for the numerator, and one for the denominator), and (iii)~two variables for each entry of $\bd$. 

Each of the $n$ divisibility constraints of $D(B\by+\bc)+\bd\in\Z^{n}$ and each of the atoms $\psi_i$ can be written in the forms \eqref{bepa-atoms}. To see this, let $u_1,\ldots,u_k$ be the bounded variables used for the numerators or denominators in $D$ and $\bd$. Observe that the vector $B\by$ is a linear combination of $\by$ with integer coefficients. The matrix $D$ and the vector $\bd$ consist of quotients of bounded variables, hence rational functions in $\Z[u_1,\ldots,u_k]$. Thus, the vector $D(B\by+\bc)+\bd$ has in each entry an expression $s+\sum_{i=1}^m r_iy_i$, where $r_1,\ldots,r_m,s\in\Z[u_1,\ldots,u_k]$. Hence, by multiplying with the product of all denominators, we can write each inequality $\ba_i^\top(D(B\by+\bc)+\bd)\le\bb_i^\top\by+c_i$ in the form of \eqref{bepa-atoms}. Moreover, for the requirement $D(B\by+\bc)+\bd\in\Z^{n}$, we can write each row of $D(B\by+\bc)+\bd$ as a quotient $\tfrac{1}{q}(r+\sum_{i=1}^m p_iy_i)$, where $q,r,p_1,\ldots,p_m\in\Z[u_1,\ldots,u_k]$, so that membership in $\Z$ is equivalent to $\sum_{i=1}^m p_iy_i\equiv -r\pmod{q}$.

Let us argue why \eqref{bepa-formula} is equivalent to
$\exists\bx\colon\varphi(\bx,\by)$. Clearly, if \eqref{bepa-formula} is satisfied, then
$\bz=(z_1,\ldots,z_\ell)$ yields a set of atoms that, if satisfied, makes
$\varphi$ true. Moreover, the vector $D(B\by+\bc)+\bd$ is an integer vector
that satisfies all the atoms specified by $\bz$. 

Conversely, suppose
$\varphi(\bx,\by)$ holds for some $\bx\in\Z^n$ and $\by\in\Z^m$. First, we set exactly those $z_i$ to $1$ for which the $i$-th atom in $\varphi$ is satisfied by $\bx,\by$. Recall that each row of $A$ (and each row of $B$, and of $\bc$) corresponds to an atom in $\varphi$. Let $A'$ be the matrix obtained from $A$ by selecting those rows that correspond to atoms that are satisfied by our $\bx$ and $\by$. Define $B'$ similarly from $B$, and $\bc'$ from $\bc$. Then we have $A'\bx\le B'\by+\bc'$. Now \cref{solution-affine-transformation} yields a matrix $D'$ and a vector $\bd$ (each with $n$ rows) with $A'(D'(B'\by+\bc')+\bd)\le B'\by+\bc'$. Now the set of rows of $B'\by+\bc'$ is a subset of the rows of $B\by+\bc$, so by inserting zero-columns into $D'$, we can construct a matrix $D$ with $D(B\by+\bc)=D'(B'\by+\bc')$. Hence, we have $A'(D(B\by+\bc)+\bd)\le B'\by+\bc'$. The latter means exactly that $\psi_i$ is satisfied for every $i$ with $z_i=1$.
Thus, this choice of $z_1,\ldots,z_\ell$, $D$, and $\bd$ satisfies \eqref{bepa-formula}. This establishes \cref{bounded-existential}.

\subsection{Constructing solutions as affine transformations}
\subsubsection{Convex geometry}
Before we start with the proof of \cref{solution-affine-transformation},
we recall some standard definitions from convex geometry from Schrijver's
book~\cite{Schrijver1986}. Below, we let $\R_+=\{ r\in \R\:|\: r \ge 0 \}$.
A \emph{polyhedron} is a set $P=\{\bx \in \R^n \:|\:A\bx\leq \bb\}$, where $A$ is
an $\ell \times n$ integer matrix and $\bb \in \Z^\ell$. 
Let $C \subseteq \R^\ell$, then
$C$ is a \emph{convex cone} if $\lambda \bx+\mu \by\in C$ for all $\bx,\by\in C$ and $\lambda,\mu \in \R_+$. Given a set $X \subseteq \R^\ell$, 
\[\cone(X)=\{\lambda_1\bx_1+\dots+\lambda_t\bx_t\mid t\geq 0,~\bx_1,\dots,\bx_t\in X,~\lambda_1,\dots,\lambda_t\in \R_+ \}\,.\]
The     \emph{convex hull} of a set $X\subseteq \R^\ell$ is the smallest convex set containing that set, i.e.,
\begin{multline*}
	\convhull(X)=\{\lambda_1\bx_1+\dots +\lambda_t\bx_t\mid t\geq 1,~\bx_1,\bx_2,\dots \bx_t\in X, \\
\lambda_1,\dots,\lambda_t\in \R_+,~\lambda_1+\dots+\lambda_t=1\}\,.
\end{multline*}

Next, we recall some terminology concerning the structure of polyhedra. 
The \emph{characteristic cone} of a polyhedron $P=\{\bx \mid A\bx \leq \bb\} \subseteq \R^n$ is the set $\charcone(P)\coloneqq \{\by \in \R^n \mid A\by\leq 0\}$.
The \emph{lineality space} of polyhedron $P$ is the set $\linspace(P)\coloneqq \{\by \in \R^n \mid A\by=0\}$.

\begin{definition}[Faces]
	Given a polyhedron $P \subseteq \R^n$, $F\subseteq P$ is a \emph{face} of $P$ if and only if $F$ is non-empty and 
	\[F=\{\bx\in P\:|\: A^{\prime}\bx=\bb^{\prime}\}\]
	for some subsystem $A^{\prime}\bx\leq \bb^{\prime}$ of $A\bx\leq \bb$. We call
	$F \subseteq P$ a \emph{proper face} of $P$
	if $F \neq \emptyset$ and $F \neq P$.
\end{definition}
It follows that $P$ has only finitely many faces. A \emph{minimal face} of $P$ is a face not containing any other face.
We have the following characterization of minimal faces \cite[Thm~8.4]{Schrijver1986}, 
\begin{proposition}
	\label{char-min-face}
	A set $F$ is a minimal face of a polyhedron $P \subseteq \R^n$ if and only if $\emptyset \neq F\subseteq P$ and 
	\[F=\{\bx \in \R^n \:|\:A^{\prime}\bx=\bb^{\prime}\}\]
	for some subsystem $A^{\prime}\bx\leq \bb^{\prime}$ of $A\bx\leq \bb$, such that the matrix $A^{\prime}$ has the same rank as $A$.
\end{proposition}

The following is shown in \cite[Sec.~8.8]{Schrijver1986}:
\begin{proposition}\label{minimal-proper-faces}
	Let $C$ be the cone $\{\bx\in\R^n \:|\: A\bx \leq \bzero\}$. There is a finite collection $G_1,G_2.\dots, G_s$ of subsets, which are of the form $G_i=\{\bx\in\R^n \mid \ba_i^\top\bx\le 0,~A'\bx=\bzero\}$, where $\begin{bmatrix} A' \\ \ba_i^\top\end{bmatrix}$ is a subset of the rows of $A$, such that the following holds. If we choose for each $i=1,\ldots,s$ a vector $\by_i$ from $G_i\setminus\linspace(C)$ and choose $\bz_0,\ldots,\bz_t$ in $\linspace(C)$ such that $\linspace(C)=\cone(\bz_0,\ldots,\bz_t)$, then
		\[ C=\cone(\by_1,\ldots,\by_s,\bz_0,\ldots,\bz_t). \]
\end{proposition}
Here, the sets $G_i$ are also called minimal proper faces (but
\cref{minimal-proper-faces} is not a characterization of those).

\subsubsection{Proof of \cref{solution-affine-transformation}}
We now prove \cref{solution-affine-transformation}. For the remainder of the
section, let  $A\in\Z^{\ell\times n}$ and $\bb\in\Z^{\ell}$. Moreover, let $\Delta$
be an upper bound on all absolute values of the sub-determinants of $A$. 
Our first step is a simple application of standard facts about polyhedra.
\begin{lemma}\label{solution-rational-affine-transformation}
If the system $A\bx\le \bb$ has a solution in $\Q^{n}$, then it has one of the form $\frac{1}{a}E\bb$, where $E\in\Z^{n\times\ell}$, $a\in\Z\setminus\{0\}$, $|a|\le\Delta$, and $\norm{E}\le\Delta$.
\end{lemma}
\begin{proof}
	It is well-known that if $A\bx\le\bb$ has a rational solution, then there is a solution inside a minimal face of the polyhedron $P=\{\bx\in\R^n\mid A\bx\le\bb\}$ defined by the system of linear inequalities $A\bx\le\bb$~\cite[Thm.~8.5]{Schrijver1986}. Recall that a \emph{minimal face} is a non-empty subset $F\subseteq P$ of the form 
\begin{equation} F=\{\bx\in\R^{n} \mid A'\bx=\bb'\}, \label{minimal-face}\end{equation}
	where $A'\bx\le\bb'$ is a subset of the inequalities in $A\bx\le\bb$ such that the matrix $A'$ has the same rank as $A$~(see \cref{char-min-face} or \cite[Thm.~8.4]{Schrijver1986}). Suppose $F$ is a non-empty minimal face and satisfies \eqref{minimal-face}. Here, we may assume that the rows of $A'$ are linearly independent (otherwise, we can remove redundant rows without changing $F$). This means, $A'$ can be written as $A'=(B\quad C)$ such that $B$ is invertible. Then the vector $\bx^*:=(B^{-1}\bb'\quad\bzero)^\top$ belongs to $F$. Since $F\subseteq P$, we know that $A\bx^*\le\bb$. By Cramer's rule, the entry $(j,i)$ of $B^{-1}$ is $\tfrac{(-1)^{i+j}\det(B_{ij})}{\det(B)}$, where $B_{ij}$ is the matrix obtained from $B$ by removing the $i$-th row and $j$-th column. Note that $|\det(B_{ij})|\le\Delta$ and $|\det(B)|\le\Delta$. In particular, $\bx^*$ can be written as $\tfrac{1}{a}E\bb$, where $a=\det(B)$ and $\norm{E}\le\Delta$.
\end{proof}

We also employ the following well-known fact, which again uses standard arguments.
\begin{lemma}\label{cone-generators}
There are integral vectors $\by_1,\ldots,\by_s$ with each component being at most $\Delta$ in absolute value, such that
$\{\bx\in\R^n \mid A\bx\le \bzero \} = \cone(\by_1,\ldots,\by_s)$.
\end{lemma}
\begin{proof}
	Let $C=\{\bx\in\R^n \mid A\bx\le\bzero\}$. The \lcnamecref{cone-generators} follows from \cref{minimal-proper-faces}. First, it is a consequence of Cramer's rule that we can choose $\bz_0,\ldots,\bz_t$ as a basis of $\linspace(C)=\{\bx\in\R^n\mid A\bx=\bzero\}$ so that all $\bz_0,\ldots,\bz_t$ are integral and have absolute values at most $\Delta$ in all components. For example, see \cite[Cor.~3.1c]{Schrijver1986}. It remains to pick from each set 
\[ G_i\setminus\linspace(C)=\{\bx\in\R^n \mid \ba_i^\top \bx< 0,~A'\bx=\bzero\} \]
an integral vector with all components bounded by $\Delta$. For
this, we can proceed similarly to
\cref{solution-rational-affine-transformation}. As a subset of rows of
$A$, the matrix $B=\begin{bmatrix} A' \\ \ba_i^\top\end{bmatrix}$ has rank
at most $n$, and we may assume $\ba_i\ne\bzero$ (otherwise
$G_i\setminus\linspace(C)$ would be empty).  Moreover, we may
assume that the rows of $B$ are linearly independent, as
otherwise we can remove rows from $A'$ without changing $G_i$.
We can thus write $B=(E\quad F)$, where $E$ is invertible. By
	Cramer's rule (see, e.g. \cite[Sec.~3.2]{Schrijver1986}), the $j$-th component of the vector
$\by=E^{-1}(0,\ldots,0,-1)$ can be written as
$\frac{1}{\det(E)}\det(\tilde{E})$, where $\tilde{E}$ is obtained from
$E$ by replacing the $j$-th column by $(0,\ldots,0,-1)$. This means, the vector
	$|\det(E)|\cdot\by$ has only integer components and all of them have
	absolute value at most $\Delta$. Now let $\by^*$ be the vector obtained
	from $|\det(E)|\cdot\by$ adding as many $0$'s as $F$ has columns. Then
	we have $B\by^*=E(|\det(E)|\cdot\by)=(0,\ldots,0,-|\det(E)|)$ and thus $\by^*\in
	G_i\setminus\linspace(C)$.
\end{proof}

We also rely on the well-known theorem of Carath\'{e}odory~\cite[Cor.~7.1(i)]{Schrijver1986}.
\begin{theorem}[Carath\'{e}odory's theorem]
If $X\subseteq\R^n$ is some subset and $\bx\in \cone(X)$, then there are linearly independent $\bx_1,\ldots,\bx_m\in X$ with $\bx\in \cone(\bx_1,\ldots,\bx_m)$.
\end{theorem}

The following lemma is the key ingredient for proving \cref{solution-affine-transformation}. Its
proof closely follows the ideas of \cite[Thm.~17.2]{Schrijver1986}, which Schrijver attributes to Cook, Gerards, Schrijver, and Tardos~\cite{cook1986tardos}. The
latter shows that for every rational $\bx$ that maximizes an expression
$\bc^\top\bx$ among the solutions of $A\bx\le\bb$, there is a close-by
integral vector that maximizes this expression among all integral vectors.
\begin{lemma}\label{solution-integral-close}
Suppose the system $A\bx\le \bb$ has an integral solution, and let $\br\in\Q^n$ be a rational solution. Then there is an integral solution $\bz^*\in\Z^n$ with $\norm{\bz^*-\br}\le n\Delta$. 
\end{lemma}
\begin{proof}
  An illustration of the proof is given in \Cref{fig:integral-close-diagram}.
Let $\bz$ be an integral solution to $A\bx\le\bb$. Split the equations $A\bx\le \bb$ into $A_1\bx\le \bb_1$ and $A_2\bx\le\bb_2$ such that 
$A_1\br\le A_1\bz$ and $A_2\br\ge A_2\bz$. In other words, we split $A$, $\bb$ into two sets of rows, depending on in which coordinates $\br$ resp.\ $\bz$ is larger. Now consider the cone $C=\{\bx\in\R^n \:\mid\: A_1\bx\ge\bzero,~A_2\bx\le\bzero\}$. Then, by the choice of $A_1$ and $A_2$, we have $\bz-\br\in C$ and therefore
\[ \bz-\br = \lambda_1\by_1+\cdots+\lambda_t\by_t, \]
where $\lambda_1,\ldots,\lambda_t\ge 0$ are real numbers and $\by_1,\ldots,\by_t$ are some linearly independent vectors chosen from the set of integer vectors $\{\by_1,\ldots,\by_s\}$ provided by \cref{cone-generators} satisfying $C=\cone(\by_1,\ldots,\by_s\}$. The choice of linearly independent vectors is possible due to Carath\'{e}odory's theorem. In particular, each $\by_i$ has maximal absolute value at most $\Delta$ and we have $t\le n$. 

Observe that for any $\mu_1,\ldots,\mu_t$ with $0\le \mu_i\le\lambda_i$ for $i\in[1,t]$, the vector
\[ \br+\mu_1\by_1+\cdots+\mu_t\by_t \]
is still a solution to $A\bx\le\bb$. Indeed, $A_1\by_i\ge\bzero$ and $A_2\by_i\le\bzero$ implies 
\begin{align*}
A_1(\br+\mu_1\by_1+\cdots+\mu_t\by_t)&\le A_1\bz\le \bb_1,~\text{and} \\
A_2(\br+\mu_1\by_1+\cdots+\mu_t\by_t)&\le A_2\br\le \bb_2, 
\end{align*}
and thus $A(\br+\mu_1\by_1+\cdots+\mu_t\by_t)\le\bb$. In particular, the vector
\[ \bz^*=\br+(\lambda_1-\lfloor\lambda_1\rfloor)\by_1+\cdots+(\lambda_t-\lfloor\lambda_t\rfloor)\by_t \]
is a solution to $A\bx\le\bb$. Moreover, $\bz^*$ is obtained from $\bz$ by subtracting integer multiples of the integer vectors $\by_1,\ldots,\by_t$, and thus $\bz^*$ is integral as well. Finally, we have
\[ \norm{\bz^*-\br}=\norm{(\lambda_1-\lfloor\lambda_1\rfloor)\by_1+\cdots+(\lambda_t-\lfloor\lambda_t\rfloor)\by_t}\le \sum_{i=1}^t\norm{\by_i}\le n\Delta. \qedhere\]
\end{proof}
\begin{figure}
	\centering
	\scalebox{0.6}[0.55]{
		\begin{tikzpicture}[x=1cm, y=1cm]
			\def\x{2};
			\draw[step=5mm, line width=0.2mm, black!40!white, semitransparent] (0,0) grid (15,12);
			\draw[step=5cm, line width=0.5mm, black!50!white, semitransparent] (0,0) grid (15,12);
			\draw[step=1cm, line width=0.3mm, black!90!white, semitransparent] (0,0) grid (15,12);
			\draw [black,fill=black,ultra thick] (0,0) circle (4pt);
			\node[below] at (0,0) {$(0,0)$};
			
			\node[draw,circle=2pt, fill=black, label=below  right:\Large{$\by_1$}](y1) at (0,3) {};
			
			\node[draw,ultra thick, circle, red, scale = 1.5](y1) at (0,3) {};
			
			\node[draw,ultra thick, circle, red, scale = 1.5](y1) at (1,0) {};
			
			\node[draw,circle=2pt, fill=black, label=below  right:\Large{$\by_2 = \by_t$}](y2) at (1,0) {};

			\node[draw,circle=2pt, fill=black, label=below  right:\Large{$\by_3$}](y3) at (5,1) {};

			\node[draw,circle=2pt, fill=black, label=below  right:\Large{$\by_{s-1}$}](ys-1) at (1,5) {};

			\node[draw,circle=2pt, fill=black, label=below  right:\Large{$\by_s$}](ys) at (2,3) {};
			
			\node[draw,circle=2pt, fill=black, label=below  right:\Large{$\bz-\br$}](zr) at (5,4.5) {};

			\draw [line width=1mm, blue] (\x+5,5) -- (\x+5,10);
			
			\draw [line width=1mm, blue, dashed] (\x+5,10) -- (\x+10,11);
			\draw [line width=1mm, blue, dashed] (\x+12,10) -- (\x+10,11);
			
			\draw [line width=1mm, blue] (\x+5,5) -- (\x+10,5);
			
			\draw [line width=1mm, blue] (\x+10,5) -- (\x+12,10);

			\node[draw,circle=2pt, fill=black, label=below  right:\Large{$\bz$}](z) at (\x+10,10) {};
			\node[draw,circle=2pt, fill=black, label=left:\Large{$\br$}](r) at (\x+5,5.5) {};

			\node[draw,circle, fill=red, label=left:{$\by_1'$}](zy1) at (\x+5,8.5) {};

			\node[draw,circle, fill=red, label= right:{$\by_2'$}](zy2) at (\x+6,5.5) {};
			
			\node[draw,circle=2pt, fill=green, label= right:\Large{$\bz^*$}](zs) at (\x+10,7) {};

			\draw[->,line width=1mm, red](r)--(zy1);

			\draw[->,line width=1mm, red](r)--(zy2);
			
			\draw[->,line width=1mm, black](r)--(z) node[draw=none,fill=none,font=\large, midway,above, sloped] {$\bz-\br = \lambda_1 \by_1 + \lambda_2 \by_2$};
			
			\draw[->,line width=1mm, green, text = black](r)--(zs) node[draw=none,fill=none,font=\large, above right, sloped, pos = 0.25] {$\bz^*-r{=}\{\lambda_1\} \by_1{+}\{\lambda_2\} \by_2$};
			
	\end{tikzpicture}}
	
	\caption{The main idea behind \cref{solution-integral-close}. The region enclosed by blue lines depicts the solution space of the given system of linear inequalities. As mentioned in the lemma, $\br$ and $\bz$ are respectively the given rational and integral solutions. Due to \cref{cone-generators}, we know that $C$ (containing $\bz-\br$) can be obtained as a cone of integer vectors $\by_1, \ldots \by_s$. Moreover, by Carath\'{e}odory's theorem, we know that there are $t$ linearly independent ($t=2$ in this case) vectors whose cone contains $\bz-\br$. Intuitively, these vectors ($\by_1, \by_2$) form a coordinate system for searching the required $z^*$. For $i\in \{1,2\}$, $\by_i' = \by_i+\br$, $\{\lambda_i\} = \lambda_i - \lfloor \lambda_i \rfloor$.} 
	\label{fig:integral-close-diagram}
\end{figure}

\begin{proof}[Proof of \cref{solution-affine-transformation}]
According to \cref{solution-rational-affine-transformation}, there is a
	rational solution $\tfrac{1}{a}E\bb$ to $A\bx\le\bb$, where
	$E\in\Z^{n\times\ell}$, $a\in\Z\setminus\{0\}$, $|a|\le\Delta$, and
	$\norm{E}\le\Delta$. We set $D:=\tfrac{1}{a}E$. Now since $A\bx\le\bb$
	has an integral solution, \cref{solution-integral-close} yields an
	integral solution $\bz^*$ close to $D\bb$, meaning
	$\norm{\bz^*-D\bb}\le n\Delta$. We set $\bd:=\bz^*-D\bb$. Then of
	course $D\bb+\bd=\bz^*$ is an integral solution to $A\bx\le\bb$.
	Moreover, we clearly have $\norm{\bd}\le n\Delta$. It remains to show
	that even $\ratnorm{\bd}\le n\Delta^2$. 
	Indeed, since $\bz^*$ is integral, $\bb$ is integral, and
	$D=\tfrac{1}{a}E$ with integral $E$, we know that in $\bd=\bz^*-D\bb$,
	every entry can be written with $a$ as its denominator. As this fraction
	has absolute value at most $n\Delta$ and $|a|\le\Delta$, both numerator
	and denominator have absolute value at most $n\Delta^2$.
\end{proof}
\section{Matching complexity lower bounds}\label{lower-bounds}

In this section we prove the lower bounds for \cref{detecting-wqos,monadic-decomposability-epa}.

\subparagraph{Detecting WQOs} We begin with the lower bound for \cref{detecting-wqos}. That is, we show that deciding whether an existential Presburger formula 
defines a WQO is $\coNEXP$ hard.  The idea is essentially the same as the $\coNP$ lower bound for detecting WQOs for quantifier-free formulas in \cite[Sec.~8]{bergstraesser2023ramsey}. The proof follows from reducing the satisfiability problem for $\Pi_2$ sentences to WQO-definability of existential Presburger formulas.  
Given an instance $\gamma$ of a  $\Pi_2$ sentence, we synthesize an existential Presburger formula $\varphi$  and 
show that $\varphi$ defines a WQO iff $\gamma$ is satisfiable.  The $\coNEXP$-completeness of $\Pi_2$ sentences follows from the $\NEXP$-completeness of  $\Sigma_2$ sentences  \cite{DBLP:conf/csl/Haase14}.

Consider an instance of a $\Pi_2$ sentence
 \begin{equation*} \gamma\coloneqq \forall \by\colon \exists \bx\colon \psi(\bx,\by)\end{equation*} 
where $\psi$ is quantifier-free, $\bx$ ranges over $\Z^n$, and $\by$ ranges over $\Z^m$. The goal is to construct an existential PA formula $\varphi$ such that $\varphi$ defines a WQO iff $\gamma$ is true. 
First we define  the formula \begin{equation*} \Gamma(\by)\coloneqq \exists \bx \psi(\bx,\by)\end{equation*}  Now, define the existential Presburger formula  $\varphi$ as follows. 
\begin{align*}
    \varphi((x,\bx),(y,\by))\coloneqq &(x<0\wedge y<0)\vee(x>0\wedge y>0)\vee (x<0\wedge y=0)\\
    \vee &(x=0\wedge y>0)\vee (x=0\wedge y=0) \\\vee &(x<0\wedge y>0\wedge \Gamma(\by)).
\end{align*}
Here, both $\bx$ and $\by$ range over $\Z^m$, hence $\varphi$ defines a relation in $\Z^{1+m}\times\Z^{1+m}$.
Since the existential quantifiers of $\Gamma$ can be moved in front of $\varphi$, 
 $\varphi$ is an existential Presburger instance. 
\begin{lemma}
    $\varphi$ defines a WQO if and only if $\gamma$ is true i.e. $\Gamma(w)$ is true $\forall \:\bw\in\Z^m$.
\end{lemma}
\begin{proof}
    ($\Rightarrow$) Let $\varphi$ define a WQO. Assume for contradiction there exists $\bw\in\Z^m$ such that $\Gamma(\bw)$ is false. Notice that, by definition, $\varphi((-1,\bw),(0,\bw))$ and $\varphi((0,\bw),(1,\bw))$ are true. By transitivity, we must have that $\varphi((-1,\bw),(1,\bw))$ is true. Therefore, $\Gamma(\bw)$ must be true. This is a contradiction.\\
    ($\Leftarrow$) Let $\Gamma(\bw)$ be true for all $\bw\in\Z^m$. Let $A$, $B$ and $C$ be sets of all vectors over $\Z^{1+m}$ with negative, zero and positive first component, respectively. It is easy to see that $\varphi$ relates all vectors within each of $A$, $B$ and $C$. Further, $\varphi(\bu,\bv)$ is true if
	\begin{itemize}
        \item $\bu\in A $ and $\bv\in B$, or
        \item $\bu\in B $ and $\bv\in C$, or
        \item $\bu\in A $ and $\bv\in C$.
    \end{itemize}
This means that $\varphi$ must be a transitive, reflexive relation. Hence, $\varphi$  trivially defines a WQO:  in any infinite sequence $\bu_1, \bu_2, \dots$ of vectors over $\Z^{1+m}$, we can always find $\bu_i, \bu_j$ with $i < j$ such that both $\bu_i,\bu_j$ belong to either $A$ or $B$ or $C$.  Since $\varphi$ relates all vectors within each of these, the lemma follows. 
\end{proof}

\subparagraph{Monadic decomposability}
Let us now show the lower bound for \cref{monadic-decomposability-epa}, i.e., that monadic decomposability for \EPA\ formulas is $\coNEXP$-hard. The idea is the same as the $\coNP$-hardness for quantifier-free formulas in \cite{DBLP:journals/pacmpl/BergstrasserGLZ24}\footnote{As Anthony W.\ Lin and Matthew Hague explained to us, it would also not be difficult to adapt the idea of the $\coNP$ lower bound in \cite[Lem.~2]{DBLP:conf/cade/HagueLRW20}.}. We reduce from the $\Pi_2$-fragment of Presburger arithmetic, which is known to be $\coNEXP$-complete (see the discussion around \cref{nexp-sigma2}). Suppose we are given a $\Pi_2$-formula $\varphi=\forall\bx\exists\by\colon \psi(\bx,\by)$, where $\bx$ contains $n$ variables, and $\by$ contains $m$ variables. We claim that the existential formula $\kappa=\exists\by\colon\psi(\bx,\by)\vee z_1=z_2$ (which has free variables $\bx$, $z_1$, $z_2$) is monadically decomposable if and only if $\varphi$ holds (see \cref{lower-bounds}), which would clearly complete the reduction.

Indeed, if $\varphi$ holds, then $\kappa$ is satisfied for every vector in $\Z^{n+2}$ and is thus clearly monadically decomposable. Conversely, if $\varphi$ does not hold, then there is some $\ba\in\Z^n$ so that $\exists\by\colon\psi(\ba,\by)$ fails to hold. If $\kappa$ were monadically decomposable, then so would the formula $\kappa\wedge \bx=\ba$, but this is equivalent to $z_1=z_2$, which is clearly not monadically decomposable. This establishes the claim and hence $\coNEXP$-hardness.

\section{An exponential lower bound for quantifier elimination}\label{exponential-lower-bound}
\newcommand{\semantics}[1]{\llbracket #1 \rrbracket}
\newcommand{\period}[1]{|#1|_{\mathsf{p}}}
Our main results show that one can eliminate a block of existential quantifiers
with only an exponential blow-up. 
Using an example from \cite[Thm.~2]{DBLP:conf/csl/Haase14}, we will now prove an exponential lower bound, even if constants are encoded in binary.

In the presence of binary encoded constants, we cannot use Weispfenning's lower
bound argument~\cite[Thm.~3.1]{DBLP:conf/issac/Weispfenning97} (even for a
singly exponential lower bound), which compares norms of vectors in finite
sets defined by $\EPA$ vs. quantifier-free formulas. Indeed, it is a simple
consequence of Pottier's bounds on Hilbert bases~\cite{DBLP:conf/rta/Pottier91}
that finite sets defined by $\EPA$ formulas consist of at most exponentially large
vectors. With binary encoded constants, one easily constructs
quantifier-free formulas defining finite sets of exponentially large vectors.

Instead, we measure the periodicity of infinite sets. Recall that every
Presburger formula with one free variable defines an \emph{ultimately periodic}
set $S\subseteq\Z$, meaning that there are $n_0,p\in\N$, $p\ge 1$, such that
for every $n\in \Z$, $|n|\ge n_0$, we have $n+p\in S$ if and only if $n\in S$.
Such a $p$ is called a \emph{period} of $S$.  For a formula $\varphi$ with one
free variable, we denote by $\period{\varphi}$ the \emph{smallest} period of
the set defined by $\varphi$. In \cite[Thm.~2]{DBLP:conf/csl/Haase14}, Haase
constructs\footnote{See also \cref{appendix-period}.} a sequence
$(\Phi_n(x))_{n\ge 0}$ of $\EPA$ formulas of size $O(n^2)$ such that
$\period{\Phi_n}$ is at least $2^{2^{\Omega(n)}}$.  The following will imply that
the formulas $\Phi_n$ require exponential-sized quantifier-free equivalents:
\begin{lemma}
	Let $\varphi$ be quantifier-free with one free variable. Then
	$\period{\varphi}\le 2^{|\varphi|}$.
\end{lemma}
\begin{proof}
	We prove this by structural induction. If $\varphi$ is an atom $ax\le
	b$, then $\period{\varphi}=1$. If $\varphi$ is an atom $ax\equiv
	b\bmod{c}$ with constants $a,b,c$ written in binary, then
	$\period{\varphi}\le |c|\le 2^{|\varphi|}$.  Moreover,
	$\period{\neg\varphi}=\period{\varphi}$. Now observe that if
	$S_1,S_2\subseteq\Z$ are ultimately periodic sets, then we have
	$\period{S_1\cup S_2}\le \period{S_1}\cdot\period{S_2}$ and
	$\period{S_1\cap S_2}\le\period{S_1}\cdot\period{S_2}$.  This implies
	$\period{\varphi_1\vee\varphi_2}\le
	\period{\varphi_1}\cdot\period{\varphi_2}\le
	2^{|\varphi_1|+|\varphi_2|}\le 2^{|\varphi|}$ and similarly
	$\period{\varphi_1\wedge\varphi_2}\le
	\period{\varphi_1}\cdot\period{\varphi_2}\le
	2^{|\varphi_1|+|\varphi_2|}\le 2^{|\varphi|}$.
\end{proof}
Now indeed, if $(\varphi_n)_{n\ge 0}$ is a sequence of quantifier-free
equivalents of $(\Phi_n)_{n\ge 0}$, then for some constant $c>0$ and large $n$,
we have $2^{|\varphi_n|}\ge \period{\varphi_n}=\period{\Phi_n}\ge 2^{2^{cn}}$
and hence $|\varphi_n|\ge 2^{cn}$.

\label{beforebibliography}
\newoutputstream{pages}
\openoutputfile{main.pages.ctr}{pages}
\addtostream{pages}{\getpagerefnumber{beforebibliography}}
\closeoutputstream{pages}
\bibliographystyle{plainurl}
\bibliography{bibliography}
\newpage 

\appendix
\section{More Details for Lemma \ref{cone-generators}}
We recall Cramer's rule which has been used in the proof. 
\begin{proposition}[Cramer's rule]
    Let a system of $n$ linear equations for $n$ unknowns be represented as
\[
    Ax=b\,,
\]
where $A$ is an invertible $(n\times n)$ matrix. This system has as unique solution given by $x=(x_1,x_2,\dots,x_n)$ where, \[
    x_i=\frac{\det(A_i)}{\det(A)}
\]
$A_i$ is the matrix formed by replacing the $i$th column of $A$ by $b$.
\end{proposition}

\section{Sets with large periods}\label{appendix-period}
The formula $\Phi_n$ constructed by Haase in \cite[Thm.~2]{DBLP:conf/csl/Haase14} defines the set
\[ S_n =\{a\in\N \mid \exists b\colon 1<b<2^n,~\text{$b$ divides $a$}\}. \]
and Haase argues that the smallest period of $S_n$ is $2^{2^{\Omega(n)}}$. While the latter is true, the argument in \cite{DBLP:conf/csl/Haase14} does not quite show this. 
The proof of \cite[Thm.~2]{DBLP:conf/csl/Haase14} argues that the smallest period of $S_n$ is the least common multiple of the numbers $\{1,\ldots,2^n-1\}$, which is lower bounded by $2^{2^{\Omega(n)}}$ according to Nair~\cite{nair1982chebyshev}. However, as we will see, the smallest period of $S_n$ is in fact a slightly smaller number. It is still lower bounded $2^{2^{\Omega(n)}}$, but this requires a different argument.  We present a correction.

An easy fix for the result would be to instead define the set
\begin{align*} 
	S'_n &=\{a\in\N \mid \exists b\colon 1<b<2^n,~\text{$b$ does not divide $a$}\} \\
	&=\{a\in\N \mid \exists b,c\colon 1<b<2^n,~1\le c<2^n,~\text{$b$ divides $a+c$} \}, 
\end{align*}
for which a simple modification of the formulas $\Phi_n$  in
\cite{DBLP:conf/csl/Haase14} yields a polynomial-sized $\EPA$ formula
$\Phi'_n$. Moreover, the smallest period of $S'_n$ is indeed the least common
multiple of $\{1,\ldots,2^n-1\}$, and so Nair's bound would apply.

However, one can show that the smallest period of $S_n$ is indeed lower bounded by $2^{2^{\Omega(n)}}$, just not by the least common multiple of $\{1,\ldots,2^n-1\}$.
For any natural $n\in\N$, define the \emph{primorial} of $n$, in symbols $n\#$, as the product of all primes $\le n$. Thus, if $p_1,p_2,\ldots$ is the sequence of all primes in ascending order and $\pi(n)$ is the number of all prime numbers $\le n$, then
\[ n\#=\prod_{i=1}^{\pi(n)} p_i. \]

\begin{claim}\label{period-primorial}
	The smallest period of $S_n$ is $2^n\#$. 
\end{claim}
\begin{proof}
	Clearly, $2^n\#$ is a period of $S_n$: $S_n$ is the set of all numbers
	that have a prime divisor among $\{2,\ldots,2^n-1\}$, and adding or
	subtracting the product of all these primes does not change that.
	
	It remains to show that $2^n\#$ is the smallest period of $S_n$.
	Suppose $k$ is a period of $S_n$. We will show that every prime $p$
	with $1<p<2^n$ is a divisor of $k$, which will clearly establish the
	claim. Let $\{p_1,\ldots,p_\ell\}$ be the primes in
	$\{2,\ldots,2^n-1\}$. Towards a contradiction, suppose there is a prime
	$p_j$, $1\le j\le \ell$, that does not divide $k$. By the Chinese
	Remainder Theorem, the system of congruences
	\begin{align*}
		x&\equiv 1\pmod{p_i} && \text{for each $i\in\{1,\ldots,\ell\}$, $i\ne j$}, \\
		x&\equiv -k\pmod{p_j} && 
	\end{align*}
	has infinitely many solutions $a\in\N$. For each such $a$, we have
	$a\notin S_n$, because $a$ is not divisible by any $p_i$, $1\le i\le
	\ell$. However, $a+k$ is divisible by $p_j$, and thus $a+k\in S_n$.
	Therefore, $k$ cannot be a period of $S_n$.
\end{proof}

Using \cref{period-primorial}, we can now obtain the $2^{2^{\Omega(n)}}$ lower
bound for the smallest period of $S_n$. This is because equation (3.14) of
\cite{BarkleySchoenfeld1962} implies that for every $m\ge 563$, we have $m\#\ge
2^{m-1}$. In particular, for $n\ge 10$, we have $2^n\#\ge 2^{2^n-1}$. This
proves that $\period{\Phi_n}$ is lower bounded by $2^{2^{\Omega(n)}}$.

\section{Incorrect lower bounds on eliminating a block of existential quantifiers}\label{appendix-weispfenning}

We elaborate on a flaw in Weispfenning's
paper~\cite{DBLP:conf/issac/Weispfenning97} which is a
consequence of misinterpreting results from
the literature, from which he incorrectly concludes that the elimination of a
block of existential quantifiers from a formula of Presburger arithmetic
results in an inherent doubly exponential blow-up.

The main result of Section~3 of~\cite{DBLP:conf/issac/Weispfenning97}
is Theorem~3.1, which states that performing quantifier elimination on
\emph{arbitrary} formulas of Presburger arithmetic results in an
inherent triply exponential blow-up, assuming unary encoding of
numbers. To this end, Weispfenning invokes a result by Fischer and
Rabin~\cite{FR98} who showed that there exists a function $g\colon \N \to
\N$ such that for almost all $n$,
\[
g(n) \ge 2^{2^{2^{n+1}}}\,,
\]
and who gave a family of formulas $\Phi_n(x,y,z)$ of Presburger
arithmetic of size linear in $n$ such that $\Phi_n(x,y,z)$ holds if
and only if $0\le x,y,z < g(n)$ and $x\cdot y = z$. He then goes on
concluding that the smallest quantifier-free formula defining the set
$\{ z \in \Z \mid \Phi_n(1,z,z) \}$ requires a formula of size at
least $g(n)$, assuming unary encoding of numbers.

Weispfenning then continues sketching how to adapt this approach in the
presence of a bounded number of quantifier alternations. To this end,
he appeals to a result by
F\"urer~\cite{DBLP:journals/tcs/Furer82}, which states that for some constant $r>0$, one can define multiplication up to 
\begin{equation}
2^{2^{(n/a)^{ra}}}\label{furer-bound}
\end{equation}
using a formula of length $n$ and $a$ quantifier alternations.  Adapting his
line of reasoning from the general case, Weispfenning applies this to $a=1$ and
concludes that eliminating a block of existential quantifiers yields an
inherent doubly exponential blow up.  F\"urer does indeed claim the existence
of such a family in the third paragraph
in~\cite[p.~108]{DBLP:journals/tcs/Furer82}.  However, a close inspection of
F\"urer's proof reveals that these formulas are not constructed \emph{for every
$a$ and $n$}, but only \emph{for infinitely many $a$ and $n$}. More
specifically,  F\"urer supposes some given $k,m\in\N$ and constructs a formula
of length $c(mk\log k+1)$ and $2m+d$ quantifier alternations (see the seventh
paragraph in~\cite[p.~108]{DBLP:journals/tcs/Furer82}). Here, $c$ and $d$
appear to be unspecified constants. By choosing $a=2m+d$ and $n=c(mk\log k+1)$,
F\"urer's claims then yield multiplication up to \eqref{furer-bound} for a
suitable $r>0$.  In particular, F\"urer's construction does not yield the
existence of such formulas for \emph{every} $a\in\N$.

Of course, from the fact that existential Presburger arithmetic allows
for defining ultimately periodic sets with a doubly exponential
period, cf.~\Cref{appendix-period}, it is not unreasonable to believe
that this could somehow be turned into a lower bound similar to the
one claimed by Weispfenning. However, such large periods can already
be produced by an exponential intersection of divisibility constraints
and thus do not imply a doubly exponential lower bound on the formula
size after eliminating a block of existentially quantified variables.

\label{afterbibliography}
\newoutputstream{pagestotal}
\openoutputfile{main.pagestotal.ctr}{pagestotal}
\addtostream{pagestotal}{\getpagerefnumber{afterbibliography}}
\closeoutputstream{pagestotal}

\newoutputstream{todos}
\openoutputfile{main.todos.ctr}{todos}
\addtostream{todos}{\arabic{@todonotes@numberoftodonotes}}
\closeoutputstream{todos}
\end{document}